\documentclass{article}
\usepackage{macros} 
\usepackage[utf8]{inputenc}

\title{Chickens and Dukes}
\author{Carl Imbens}
\affil{Proof School, San Francisco CA}
\date{\today}

\begin{document}
\maketitle

\abstract{



Following on the King Chicken Theorems originally proved by Maurer, we examine the idea of multiple flocks of chickens by bringing the chickens from tournaments to multipartite tournaments.
As Kings have already been studied in multipartite settings, notably by Koh-Tan and Petrovic-Thomassen, we examine a new type of chicken more suited than Kings for these multipartite graphs: Dukes.
We define an $m$-Duke to be a vertex from which any vertex in a different partite set is accessible by a directed path of length at most $m$.
In analogy with Maurer's paper, we prove various structural results regarding Dukes.
In particular, we prove the existence of 3-Dukes in all multipartite tournaments, and
we conclude by proving that in any multipartite tournament, either there is a 1-Duke, three 2-Dukes, or four 3-Dukes.

}

\section{Introduction}

One of the challenges of caring for chickens is introducing them to other flocks.
Flocks of chickens are tightly knit and will naturally oppose other flocks.
Even in a peaceful introduction, flocks of chickens are bound to establish pecking orders with each other, with copious amounts of pecking on both sides.
For a period of time, most pecking is inter-flock in an effort to establish a pecking order between flocks.
Additionally, pecking might not be transitive, and there might be a chicken who is not pecked at all, or a chicken who pecks none.
We ask whether there is a way to designate a dominant chicken in this inter-flock setting.
This follows from the paper by Stephen Maurer~\cite{maurer:King-chickens} on finding dominant chickens within a single flock. 

Maurer modeled a flock of chickens with a \emph{tournament}, which is a complete graph in which every edge is oriented from one vertex to another.
In Maurer's model, the vertices are chickens, and the orientation of an edge defines which chicken pecks the other.
Specifically, Maurer studied the existence of Kings, which are chickens which have peck chains (i.e., directed paths) of length at most 2 to all other chickens.
In fact, the first theorem of Maurer's paper states that every flock of chickens has a King.
We will not restate every benefit and drawback of the chicken model for tournaments, as Maurer has already done that quite well in his original King Chicken paper.
The overall idea is that even though the model does not perfectly represent the behaviors of flocks, it is an interesting way to look at tournaments, and a great example of modelling complex systems with Graph Theory.
In his paper, Maurer further considered \(m\)-Kings, which are chickens which have peck chains of length at most \(m\) to all other chickens.

We can model multiple flocks of chickens in a \emph{multipartite tournament}, which is simply a complete multipartite graph in which every edge is oriented from one vertex to another.
As in Maurer's model, the orientation of an edge still defines which chicken pecks the other.
However, what distinguishes the multipartite model from Maurer's is that the partite sets now represent flocks, and pecking only happens between chickens in distinct flocks.
Kings can be defined in the same way Maurer did.
In fact, the existence of Kings in multipartite tournaments has been studied by Gutin~\cite{gutin:radii-of-n-partite-tournaments}, Petrovic-Thomassen~\cite{petrovic-thomassen:k-partite-Kings}, Koh-Tan~\cite{koh-tan:multipartite-Kings,koh-tan:multipartite-number-of-Kings}, Gutin-Yeo~\cite{gutin-yeo:Kings-in-semicomp-multip-digraphs} and Tan~\cite{tan:3-Kings-and-4-Kings}.
Unfortunately, in multipartite tournaments, as noted by Koh-Tan~\cite{koh-tan:bipartite-4-Kings-no-3-Kings}, there is not necessarily a King.
A simple example of this can be seen in a bipartite graph with all arcs going from the first partite set to the second; no vertex will have a directed path to other vertices in its own partite set.
Taking heed from when Maurer shifted the view of Kings from one-length directed paths to two-length directed paths, we shall introduce a new notion of dominance to ensure that dominant chickens exist.


The change we will make is a simple one: we will remove the requirement that a dominant chicken has peck chains of any length to the other chickens within its own flock.
It is important to distinguish between Kings and our new class of chicken.
As they are a step down from Kings, we will call them Dukes.
In general, we will call a chicken $d$ in a multipartite tournament an \emph{\(m\)-Duke} if, for any chicken $c$ not in $d$'s flock, there exists a peck chain of length at most $m$ from $d$ to $c$.
Note that every \(m\)-Duke is also an \((m+1)\)-Duke.
In this paper, we explore how Dukes exist in multipartite tournaments.
We will go on to show not only that in every multipartite tournament there exists a 3-Duke, but that in fact there exists either a 1-Duke, three 2-Dukes, or four 3-Dukes---see Thereom~\ref{thm:multiflock-finale}.


It is also interesting to note the relationship between Dukes and Kings.
Firstly, the difference between Dukes and Kings is noticeable only when each flock contains multiple chickens.
However, when we have multiple chickens in any flock, they become distinct. An \(m\)-King is by definition stronger than an \(m\)-Duke, but is not guaranteed to exist.
Still, we can notice something else: in a graph with no 1-Duke, any $m$-Duke is an $(m+1)$-King.
This is because for every vertex in its partite set, there is some edge directed to it from a vertex outside of its partite set, which in turn has a directed path of length at most $m$ from an $m$-Duke to it. 
The result in \cite{koh-tan:multipartite-Kings} which says, in the absence of a 1-Duke, there are at least three 4-Kings, follows from our Theorem~\ref{thm:multiflock-finale}.
In fact, Theorem~\ref{thm:multiflock-finale} is stronger.
We prove the existence of what would be either three 3-Kings, or four 4-Kings in the absence of a 1-Duke.
Finally, although a 3-Duke must be 4-King in the absence of 1-Dukes, the converse is not necessarily true; a 4-King might have a path of length four to every other vertex including the vertices in its partite set, but that would only guarantee it to be a 4-Duke, not a 3-Duke.
For an example of a 4-King which is not a 3-Duke, see Figure~\ref{fig:4king-not-3duke}.
The shaded ellipses represent flocks.

\begin{figure}
    \centering
    \begin{tikzpicture}
    
        \fill[black!10] (0*360/5:2) ellipse (0.5 and 0.4) ;
        \fill[black!10] (1*360/5:2) ellipse (0.5 and 0.4) ;
        \fill[black!10] (2*360/5:2) ellipse (0.5 and 0.4) ;
        \fill[black!10] (3*360/5:2) ellipse (0.5 and 0.4) ;
        \fill[black!10] (4*360/5:2) ellipse (0.5 and 0.4) ;
    
        \node[draw, circle, fill] (a) at (0*360/5:2) {} ;
        \node[draw, circle, fill] (b) at (1*360/5:2) {} ;
        \node[draw, circle, fill] (c) at (2*360/5:2) {} ;
        \node[draw, circle, fill] (d) at (3*360/5:2) {} ;
        \node[draw, circle, fill] (e) at (4*360/5:2) {} ;

        \draw[-latex] (a) to (b) ;
        \draw[-latex] (b) to (c) ;
        \draw[-latex] (c) to (d) ;
        \draw[-latex] (d) to (e) ;
        \draw[-latex] (e) to (a) ;
        \draw[-latex] (d) to (a) ;
        \draw[-latex] (c) to (a) ;
        \draw[-latex] (e) to (b) ;
        \draw[-latex] (e) to (c) ;
        \draw[-latex] (d) to (b) ;
    \end{tikzpicture}
    \caption{A multipartite tournament without transmitters containing a \(4\)-King (far right) which isn't a \(3\)-Duke.}
    \label{fig:4king-not-3duke}
\end{figure}
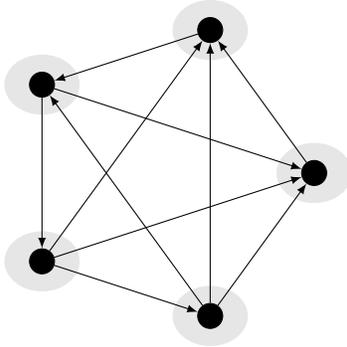

\section{Preliminary information}

In this paper, we will consider a \emph{multi-flock chicken graph} (resp. \emph{$n$-flock chicken graph} for some $n \geqslant 1$) to be a multipartite (resp. $n$-partite) graph with an orientation.
Similarly, a \emph{bi-flock chicken graph} will refer to a complete bipartite graph with an orientation.
Occasionally, we will refer to these graphs simply as \emph{multi-flock graphs} or \emph{chicken graphs}.
We will refer to the vertices of the graphs as \emph{chickens} and the partite sets as \emph{flocks}.
Keeping with standard notation, we will let \(V_1,\ldots,V_n\) denote the flocks in an \(n\)-flock chicken graph and let \(V = V_1 \cup \cdots \cup V_n\) denote the entire set of chickens.
If \(c\) is a chicken in a chicken graph, we will let \(F_c\) (rather than \(V_c\)) denote the flock containing \(c\).
If \(c\) and \(d\) are two chickens in a chicken graph, then we will use \(c \rightarrow d\) or the phrase ``\(c\) pecks \(d\)'' to denote that the edge between \(c\) and \(d\) is oriented from \(c\) to \(d\).
We will use the phrase \emph{peck chain} to denote a directed path.
In particular,
given \(m+1\) chickens \(c,c_1,c_2\ldots,c_m\) such that
\[
c \rightarrow c_1 \rightarrow c_2 \rightarrow \cdots \rightarrow c_m,
\]
we will say that \(c\) has a \emph{peck chain} of length \(m\) to \(c_m\).

We now recall our central definition.
A chicken $d$ in a multi-flock chicken graph is an \emph{$m$-Duke} if, for any chicken $c \notin F_d$, there exists a peck chain of length at most $m$ from $d$ to $c$.
For example, a \(1\)-Duke pecks every chicken not in its flock, and such chickens have been called \emph{transmitters} in the literature.
For another example, if \(d\) is a \(2\)-Duke, then every chicken not in $d$'s flock is pecked by some chicken which $d$ pecks.
As we mentioned before, in a tournament (i.e., a multipartite tournament where each partite set is a singleton), \(m\)-Dukes and \(m\)-Kings are the same.
Maurer showed that \(2\)-Dukes exist in this setting,
which leads naturally to the following question.

Does the existence of 2-Dukes generalize to chicken graphs where the flocks have multiple chickens?
The answer to this question is no.
We can quickly see that a bi-flock graph presents a clear contradiction: if there is no chicken who pecks all chickens in the other flock, there will be no 2-Duke.
We can, sadly, generalize this onto an arbitrary number of flocks, where two flocks exist as described, and the chickens in those two flocks peck all other chickens.
As this is a contradiction that 1- or 2-Dukes must exist in a multi-flock chicken graph, we will move on to 3-Dukes.
Even from the start, we can see advantages of a 3-Duke in bi-flock chicken graphs.
Since chickens only peck chickens in the other flock, any path from one chicken to the other flock will have odd length.
Is this enough of an advantage to guarantee a 3-Duke's existence?
Well, we will show that 3-Dukes exist not only in any bi-flock chicken graph, but in any multi-flock chicken graph---see Theorem~\ref{thm:nflock-exists-3Duke}.
However, this is not a strong bound, and so we will improve on it. 

We will require several more terms.
We already know what constitutes a chicken graph, and what a Duke is.
We will use the following convention throughout: chickens are represented by lower case letters, and flocks by upper case letters.
We define two flocks to be \emph{balanced} if there is no chicken in either flock which pecks all chickens in the other flock.
Conversely, we will say a flock \(V_i\) \emph{dominates} another flock \(V_j\) if there is some chicken in \(V_i\) who pecks all chickens in \(V_j\).
Furthermore, in this case, we will define a chicken \(c \in V_i\) pecking every chicken in \(V_j\) to be a \emph{dominating} chicken.
In any flock, we define a chicken who pecks at least as many chickens as any other chicken in its flock to be a \emph{prominent} chicken.
Given a chicken $d$ and a set of chickens $\mathcal{A}$, if $d$ has peck chains of length at most $m$ to every chicken in $\mathcal{A}$, we will say that $d$ is an \emph{$m$-Duke over $\mathcal{A}$}.
Finally, we will say that a chicken $e$ \emph{eclipses} another chicken $d \in F_e$ if $e$ pecks all of the chickens $d$ pecks and at least one chicken $d$ does not peck.
A chicken is \emph{non-eclipsed} if no other chicken in its flock eclipses it.

\section{Bi-flock graphs}

We are now ready to show the existence of 3-Dukes, starting in bi-flock graphs.

\begin{lemma}
\label{lemma:bi-flock-3-Duke}
In a bi-flock chicken graph, any chicken who pecks a prominent chicken must be a 3-Duke. 
\end{lemma}

\begin{proof}
Let \(k \in V_1\) be any prominent chicken, and let $d \in V_2$ be a chicken which pecks $k$.
We claim that \(d\) is a 3-Duke.
Consider any chicken $c \in V_1$.
If \(d \rightarrow c\), we are done.
Otherwise, \(c \rightarrow d\), as in Figure~\ref{fig:biflock-peck-prom}.
\begin{figure}[ht]
    \centering
    \begin{tikzpicture}
        \fill[black!10] (-2,-1) ellipse (1 and 2) ;
        \fill[black!10] (2,-1) ellipse (1 and 2) ;
        
        \node (V1) at (-2,1.5) {$V_1$} ;
        \node (V2) at (2,1.5) {$V_2$} ;
    
        \node[draw, circle] (k) at (-2,0) {$k$} ;
        \node[draw, circle] (d) at (2,0) {$d$} ;
        \node[draw, circle] (f) at (2,-2) {$f$} ;
        \node[draw, circle] (c) at (-2,-2) {$c$} ;
        
        \draw[-latex] (k) -- (f) ;
        \draw[-latex] (f) -- (c) ;
        \draw[-latex] (c) -- (d) ;
        \draw[-latex] (d) -- (k) ;
        
    \end{tikzpicture}
    \caption{$d$ pecks a prominent chicken, $k$.}
    \label{fig:biflock-peck-prom}
\end{figure}
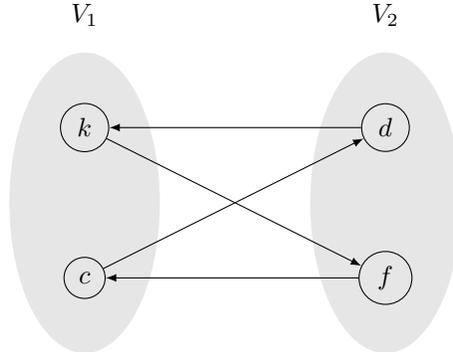
Moreover, since \(c\) does not peck more chickens than \(k\), there must be some chicken \(f \in V_2\) such that \(k \rightarrow f \rightarrow c\).
This implies $d \rightarrow k \rightarrow f \rightarrow c$, and we can conclude $d$ is a 3-Duke.
Finally, this argument still holds with the roles of \(V_1\) and \(V_2\) reversed.
\end{proof}

We know that, in an unbalanced pair of flocks, there exists a prominent chicken who is not pecked at all.
Such a chicken is clearly a 3-Duke.
In fact, it is a 1-Duke, so we are left to consider a balanced pair of flocks.

\begin{corollary}
\label{cor:balanced-bi-flock-3-Duke}
In a bi-flock graph where the two flocks are balanced, each flock contains a 3-Duke.
\end{corollary}

\begin{proof}
By the definition of ``balanced,'' a prominent chicken in either flock must be pecked by some chicken in the other flock.
Therefore,  by Lemma~\ref{lemma:bi-flock-3-Duke}, each flock contains a 3-Duke.
\end{proof}

One of the things that is further explored in Maurer's paper is the possibility of multiple Kings.
In Theorem~\ref{thm:bi-flock-1Duke-or43Dukes} below, not only will we show it is possible for there to be multiple 3-Dukes, but we will show it is a necessity in the absence of a 1-Duke.
We require one more lemma before proving the theorem.

\begin{lemma}
\label{lemma:bi-flock-pecks-most}
In any bi-flock chicken graph, a prominent chicken in a non-dominated flock is a 3-Duke.
\end{lemma}

\begin{proof}
Without loss of generality, suppose \(V_1\) is not dominated by \(V_2\), and let \(d \in V_1\) be a prominent chicken.
Consider an arbitrary chicken \(k \in V_2\).
If \(d \rightarrow k\), we are done.
Otherwise, \(k \rightarrow d\).
See Figure~\ref{fig:bi-flock-pecks-most} for a depiction of this case.
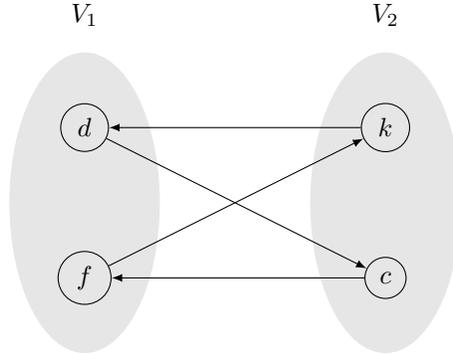
\begin{figure}[ht]
    \centering
    \begin{tikzpicture}
        \fill[black!10] (-2,-1) ellipse (1 and 2) ;
        \fill[black!10] (2,-1) ellipse (1 and 2) ;
        
        \node (V1) at (-2,1.5) {$V_1$} ;
        \node (V2) at (2,1.5) {$V_2$} ;
    
        \node[draw, circle] (d) at (-2,0) {$d$} ;
        \node[draw, circle] (k) at (2,0) {$k$} ;
        \node[draw, circle] (c) at (2,-2) {$c$} ;
        \node[draw, circle] (f) at (-2,-2) {$f$} ;
        
        \draw[-latex] (k) -- (d) ;
        \draw[-latex] (f) -- (k) ;
        \draw[-latex] (c) -- (f) ;
        \draw[-latex] (d) -- (c) ;
    \end{tikzpicture}
    \caption{The case \(k \rightarrow d\) in Lemma~\ref{lemma:bi-flock-pecks-most}.}
    \label{fig:bi-flock-pecks-most}
\end{figure}
Since \(V_1\) is not dominated by \(V_2\), there exists some chicken \(f \in V_1\) such that \(f \rightarrow k\).
Furthermore, since \(d\) is prominent, \(f\) cannot \emph{also} peck all the chickens \(d\) pecks.
Therefore, for some chicken \(c \in V_2\), we have
\(d \rightarrow c \rightarrow f \rightarrow k\).
Thus, \(d\) is a 3-Duke.
\end{proof}

\begin{theorem}
\label{thm:bi-flock-1Duke-or43Dukes}
For any bi-flock chicken graph, either there exists a 1-Duke, or there exist four 3-Dukes.
\end{theorem}

\begin{proof}
Suppose there is no 1-Duke.
Then the two flocks are balanced, and every chicken is pecked.
Let \(k_1 \in V_1\) and \(k_2 \in V_2\) be prominent chickens.
By Lemma~\ref{lemma:bi-flock-pecks-most}, $k_1$ and $k_2$ are 3-Dukes.
Furthermore, by Lemma~\ref{lemma:bi-flock-3-Duke}, all chickens pecking \(k_1\) or \(k_2\) are 3-Dukes.
If at least two chickens peck each of \(k_1\) and \(k_2\), we are done.

Otherwise, without loss of generality, we can assume there is exactly one chicken, \(d \in V_2\), who pecks \(k = k_1\).
In this case, any chicken which pecks $d$ would be a 3-Duke, as it would peck $d$ and have a peck chain of length 2 to $k$, which pecks all chickens in \(V_2\) except $d$.

If $d$ is pecked at least twice, then there exist two 3-Dukes $f_1, f_2 \in V_1$ which peck \(d\) and are not $k$, as in Figure~\ref{fig:bi-flock-thm-case1}.
Therefore, $f_1 , f_2 , k$, and $d$ are all 3-Dukes.
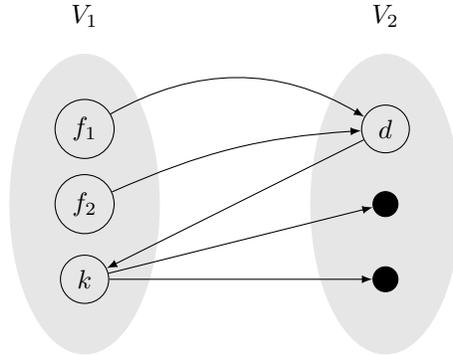
\begin{figure}[ht]
    \centering
    \begin{tikzpicture}
        \fill[black!10] (-2,-1) ellipse (1 and 2) ;
        \fill[black!10] (2,-1) ellipse (1 and 2) ;
        
        \node (V1) at (-2,1.5) {$V_1$} ;
        \node (V2) at (2,1.5) {$V_2$} ;
    
        \node[draw, circle] (d) at (2,0) {$d$} ;
        \node[draw, circle] (k) at (-2,-2) {$k$} ;
        \node[draw, circle] (f1) at (-2,0) {$f_1$} ;
        \node[draw, circle] (f2) at (-2,-1) {$f_2$} ;
        
        \node[draw, circle, fill] (a) at (2,-1) {} ;
        \node [draw, circle, fill] (b) at (2,-2) {} ;
        
        \draw[-latex] (d) -- (k) ;
        \draw[-latex] (f1) to [bend left] (d) ;
        \draw[-latex] (f2) to [bend left=10] (d) ;
        \draw[-latex] (k) -- (a) ;
        \draw[-latex] (k) -- (b) ;
    \end{tikzpicture}
    \caption{\(d\) is pecked twice in Theorem~\ref{thm:bi-flock-1Duke-or43Dukes}.}
    \label{fig:bi-flock-thm-case1}
\end{figure}

If $d$ is pecked exactly once, then there exists some $f \in V_1$ who pecks $d$ and some $g \in V_2$ who pecks $f$, as in Figure~\ref{fig:bi-flock-thm-case2}.
Given that $g \rightarrow f \rightarrow d$ and the fact that $d$ is only pecked once, we have that $g$ is a 3-Duke.
In this case, \(d,k,f\), and \(g\) are our four 3-Dukes.
\begin{figure}[ht]
    \centering
    \begin{tikzpicture}
        \fill[black!10] (-2,-1) ellipse (1 and 2) ;
        \fill[black!10] (2,-1) ellipse (1 and 2) ;
        
        \node (V1) at (-2,1.5) {$V_1$} ;
        \node (V2) at (2,1.5) {$V_2$} ;
    
        \node[draw, circle] (d) at (2,0) {$d$} ;
        \node[draw, circle] (k) at (-2,-2) {$k$} ;
        \node[draw, circle] (f) at (-2,-1) {$f$} ;
        
        \node[draw, circle] (g) at (2,-1) {\(g\)} ;
        \node [draw, circle, fill] (b) at (2,-2) {} ;
        \node [draw, circle, fill] (a) at (-2,0) {} ;
        
        \draw[-latex] (d) to [bend right=20] (a) ;
        \draw[-latex] (d) -- (k) ;
        \draw[-latex] (f) to [bend left=10] (d) ;
        \draw[-latex] (g) to (f) ;
        \draw[-latex] (k) to [bend right=20] (b) ;
        \draw[-latex] (k) to [bend right=10] (g) ;
    \end{tikzpicture}
    \caption{\(d\) is pecked once in Theorem~\ref{thm:bi-flock-1Duke-or43Dukes}.}
    \label{fig:bi-flock-thm-case2}
\end{figure}
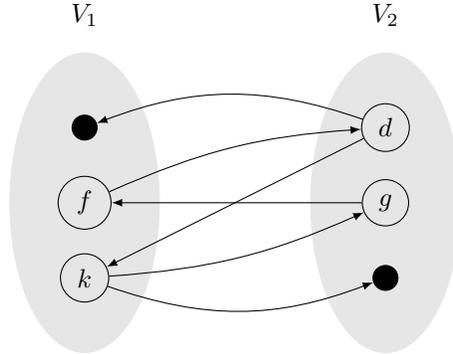
\end{proof}

We can see two cases, together demonstrating our proof is of maximal strength.
Figure~\ref{fig:biflock-1-1d} shows a bi-flock graph with a 1-Duke and no other Dukes.
Figure~\ref{fig:biflock-4-3ds} shows a bi-flock graph with exactly four 3-Dukes and no 1-Duke.

\begin{figure}[ht]
    \centering
    \begin{tikzpicture}
        \fill[black!10] (0,0) ellipse (1 and 0.5) ;
        \fill[black!10] (0,-2) ellipse (2 and 0.5) ;
        
        \node[draw,circle,fill] (a) at (0,0) {} ;
        \node[draw,circle,fill] (b) at (0,-2) {} ;
        \node[draw,circle,fill] (b-) at (-1,-2) {} ;
        \node[draw,circle,fill] (b+) at (1,-2) {} ;
        
        \draw[-latex] (a) to (b) ;
        \draw[-latex] (a) to (b-) ;
        \draw[-latex] (a) to (b+) ;
        
    \end{tikzpicture}
    \caption{A bi-flock graph with exactly one 1-Duke and no other Dukes.}
    \label{fig:biflock-1-1d}
\end{figure}
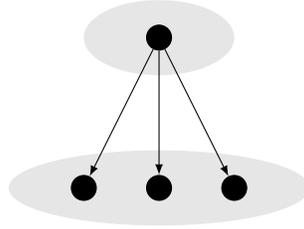

\begin{figure}[ht]
    \centering
    \begin{tikzpicture}
        \fill[black!10] (0,0) ellipse (2 and 0.5) ;
        \fill[black!10] (0,-2) ellipse (3 and 0.5) ;
        
        \node[draw,circle,fill] (a) at (-1,0) {} ;
        \node[draw,circle,fill] (b) at (1,0) {} ;
        
        \node[draw,circle,fill] (c) at (-2,-2) {} ;
        \node[draw,circle,fill] (d) at (-1,-2) {} ;
        \node[draw,circle,fill] (e) at (1,-2) {} ;
        \node[draw,circle,fill] (f) at (2,-2) {} ;
        
        \draw[-latex] (a) to (d) ;
        \draw[-latex] (a) to (e) ;
        \draw[-latex] (a) to (f) ;
        
        \draw[-latex] (b) to (c) ;
        \draw[-latex] (b) to (d) ;
        \draw[-latex] (b) to (e) ;
        
        \draw[-latex] (c) to (a) ;
        \draw[-latex] (f) to (b) ;
        
    \end{tikzpicture}
    \caption{A bi-flock graph with only four 3-Dukes and no 1-Duke.}
    \label{fig:biflock-4-3ds}
\end{figure}
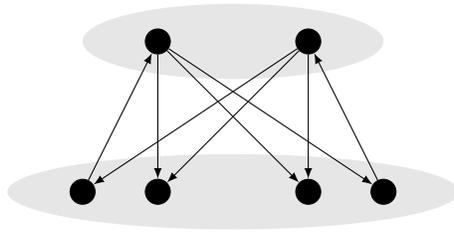

\section{Multi-flock graphs}

We begin this section by showing that a 3-Duke exists in any $n$-flock graph. 

\begin{theorem}
\label{thm:nflock-exists-3Duke}
In any multi-flock chicken graph, there exists a 3-Duke.
\end{theorem}

\begin{proof}
Consider the case in which there is some flock \(V_i\) which is balanced with or dominates every other flock.
Then we may consider the bi-flock graph with \(V_i\) as one flock and \(V \setminus V_i\) as the other flock, disregarding the edges within $V \setminus V_i$.
See Figures~\ref{fig:combine-flocksA} and \ref{fig:combine-flocksB} for depictions of this case.
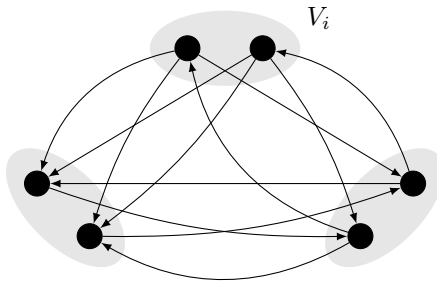
\begin{figure}[ht]
    \centering
    \begin{tikzpicture}
        \fill[black!10] (0,0) ellipse (1 and 0.5) ;
        \node (Vi) at (1.25,0.4) {\(V_i\)} ;
        
        \fill[black!10,rotate=45] (0,-3) ellipse (1 and 0.5) ;
        \fill[black!10,rotate=-45] (0,-3) ellipse (1 and 0.5) ;
        
        \node[draw,circle,fill] (a) at (-0.5,0) {} ;
        \node[draw,circle,fill] (b) at (0.5,0) {} ;
        
        \def\x{2.5}
        \def\y{-1.8}
        \def\ep{0.7}
        \pgfmathsetmacro\ox{\x-\ep}
        \pgfmathsetmacro\oy{\y-\ep}

        \node[draw,circle,fill] (c) at (\x,\y) {} ;
        \node[draw,circle,fill] (d) at (\ox,\oy) {} ;
        
        \node[draw,circle,fill] (e) at (-\ox,\oy) {} ;
        \node[draw,circle,fill] (f) at (-\x,\y) {} ;

        \draw[-latex] (a) to[bend right=10] (e);
        \draw[-latex] (a) to[bend right] (f);
        \draw[-latex] (b) to[bend left=10] (e);
        \draw[-latex] (b) to (f);
        
        \draw[-latex] (a) to (c);
        \draw[-latex] (d) to[bend left] (a);
        \draw[-latex] (c) to[bend right] (b);
        \draw[-latex] (b) to[bend left=10] (d);
        
        \draw[-latex] (e) to[bend right=10] (c);
        \draw[-latex] (c) to (f);
        \draw[-latex] (f) to[bend right=10] (d);
        \draw[-latex] (d) to[bend left] (e);
        
    \end{tikzpicture}
    \caption{\(V_i\) is balanced with or dominates the remaining flocks.}
    \label{fig:combine-flocksA}
\end{figure}
\begin{figure}[ht]
    \centering
    \begin{tikzpicture}
        \fill[black!10] (0,0) ellipse (1 and 0.5) ;
        \node (Vi) at (1.25,0.4) {\(V_i\)} ;
        \node (Vminus) at (3.5,-2.5) {\(V \setminus V_i\)} ;
        
        
        \fill[black!10] (0,-2) ellipse (3 and 1) ;
        
        \node[draw,circle,fill] (a) at (-0.5,0) {} ;
        \node[draw,circle,fill] (b) at (0.5,0) {} ;
        
        \def\x{2.35}
        \def\y{-1.8}
        \def\ep{0.7}
        \pgfmathsetmacro\ox{\x-\ep}
        \pgfmathsetmacro\oy{\y-\ep}

        \node[draw,circle,fill] (c) at (\x,\y) {} ;
        \node[draw,circle,fill] (d) at (\ox,\oy) {} ;
        
        \node[draw,circle,fill] (e) at (-\ox,\oy) {} ;
        \node[draw,circle,fill] (f) at (-\x,\y) {} ;

        \draw[-latex] (a) to[bend right=10] (e);
        \draw[-latex] (a) to[bend right] (f);
        \draw[-latex] (b) to[bend left=10] (e);
        \draw[-latex] (b) to (f);
        
        \draw[-latex] (a) to (c);
        \draw[-latex] (d) to[bend left] (a);
        \draw[-latex] (c) to[bend right] (b);
        \draw[-latex] (b) to[bend left=10] (d);

    \end{tikzpicture}
    \caption{\(V \setminus V_i\) as a single flock in Theorem~\ref{thm:nflock-exists-3Duke}.}
    \label{fig:combine-flocksB}
\end{figure}
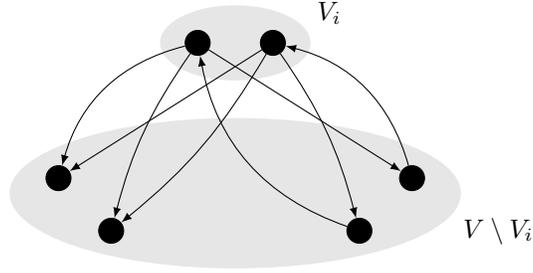
Either this pair of flocks is balanced, in which case Corollary~\ref{cor:balanced-bi-flock-3-Duke} guarantees $V_i$ contains a 3-Duke, or \(V_i\) dominates \(V \setminus V_i\) and therefore still contains a 3-Duke.
Adding the edges within \(V \setminus V_i\) back in, the 3-Duke in \(V_i\) will stay a 3-Duke.

For the remainder of the proof, consider the case in which any given flock is dominated by some other flock.
Let $V_i$ be a flock which dominates the most flocks, and let $d$ be a chicken that dominates \(V_i\).
We claim that \(d\) is a 3-Duke.

Let \(\mathcal{C}\) be the collection of flocks which $V_i$ dominates.
The remaining flocks---namely, those balanced with or dominating \(V_i\)---have one of two properties: either they do not dominate all flocks in $\mathcal{C}$, or they do.
Let \(\mathcal{A}\) be the set of remaining flocks which do not dominate every flock in \(\mathcal{C}\), and let $\mathcal{B}$ be the set of remaining flocks which dominate every flock in $\mathcal{C}$.
For any flock \(A \in \mathcal{A}\) and any chicken \(a \in A\), there exists a flock \(C \in \mathcal{C}\) and a chicken \(c \in C\) such that \(c \rightarrow a\).
No flock \(B \in \mathcal{B}\) can also dominate $V_i$, or else \(B\) would dominate more flocks than $V_i$, so \(V_i\) is balanced with every flock in \(\mathcal{B}\)

\begin{figure}[ht]
    \centering
    \begin{tikzpicture}
    
        \draw[dashed] (-2,1) ellipse (1 and 0.5) ;
        \draw[dashed] (0,-2) ellipse (2 and 0.5) ;
        \draw[dashed] (2,1) ellipse (1 and 0.5) ;
        
        \fill[black!10] (0,-2) ellipse (0.3 and 0.2) ;
        \fill[black!10] (-1,-2) ellipse (0.3 and 0.2) ;
        \fill[black!10] (1,-2) ellipse (0.3 and 0.2) ;
        
        \fill[black!10] (2,1) ellipse (0.3 and 0.2) ;
    
        \fill[black!10] (0,-0.25) ellipse (1 and 0.5) ;
        \fill[black!10] (-2,1) ellipse (0.6 and 0.4) ;
        
        \node at (0,-0.25) {\(V_i\)} ;
        \node at (2.5,-2) {\(\mathcal{C}\)} ;
        \node[draw, circle] at (-2,1) {\(d\)} ;
        \node at (-3.5,1) {\(\mathcal{A}\)} ;
        \node at (3.5,1) {\(\mathcal{B}\)} ;
        
        \draw[-latex] (-1.25,0.55) -- (-0.75,0.2) ;
        \draw[-latex] (0,-0.85) -- (0,-1.35) ;
        \draw[-latex] (-1.5,-1.35) -- (-1.7,0.4) ;
        \draw[-latex] (0.75,0.2) -- (1.25,0.55) ;
        
    \end{tikzpicture}
    \caption{The final case in Theorem~\ref{thm:nflock-exists-3Duke}.}
    \label{fig:first-thm-final-case}
\end{figure}
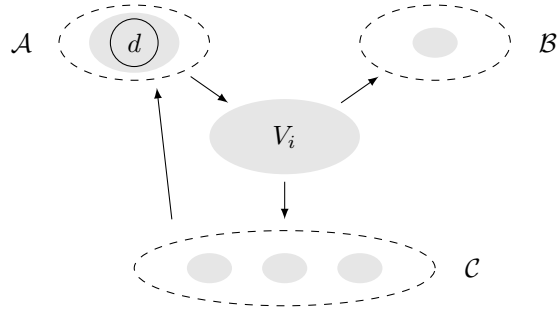

We will now show that $d$ has peck chains of length 3 or less to all chickens not in its flock.
See Figure~\ref{fig:first-thm-final-case} for a diagram illustrating these peck chains.
Since $d$ dominates $V_i$, there is a chain of length 1 from $d$ to all chickens in $V_i$.
Furthermore, as $V_i$ dominates every flock in $\mathcal{C}$, there is a  chain of length at most 2 from $d$ to every chicken in the flocks of \(\mathcal{C}\).
Any chicken in a flock of \(\mathcal{A}\) is pecked by a chicken in a flock of \(\mathcal{C}\), so there is a chain of length 3 from $d$ to all chickens in $\mathcal{A}$.
As $V_i$ is balanced with each flock in $\mathcal{B}$, there is a chain of length 2 from $d$ to all chickens in $B$, by way of \(V_i\).
Because there are peck chains of length 3 or less from $d$ to all chickens in our graph, $d$ is a 3-Duke.
\end{proof}

There we have it!
Every multi-flock chicken graph contains a 3-Duke.
We will eventually generalize Theorem~\ref{thm:bi-flock-1Duke-or43Dukes} to multi-flock chicken graphs---see Theorem~\ref{thm:multiflock-finale} below.

At the base of many of Maurer's proofs was the following result: any chicken that is pecked, is pecked by a King.
Can we prove the same for Dukes? 
The proof is not immediately obvious.
In Maurer's proof, he considered a pecked chicken \(c\) and separated the remaining chickens into two groups: chickens pecking \(c\), and chickens pecked by \(c\).
He argued that the chickens which pecked $c$ were Kings over the chickens pecked by $c$, and that they would have a King among them, maKing that King a King over all of the chickens.
Our proof is not so simple, as we have a third group: chickens in $c$'s flock.

\begin{theorem}
\label{thm:any-chicken-pecked}
Any chicken pecked is either pecked by a 3-Duke, or shares a flock with a 2-Duke.
\end{theorem}

\begin{proof}
Consider a chicken $c$, and let $\mathcal{A}$ be the nonempty set of chickens which peck $c$.
As $\mathcal{A}$ itself is a multi-flock chicken graph, $\mathcal{A}$ must contain a 3-Duke by Theorem~\ref{thm:nflock-exists-3Duke}.
Let $d$ be a 3-Duke over all of \(\mathcal{A}\) which pecks the most chickens in the overall graph.
If we can show $d$ is a 3-Duke over the entire graph, then we are done.
Otherwise, we must show that there is a 2-Duke in $c$'s flock.

We have two remaining sets of chickens to consider.
Let $\mathcal{K}$ be the set of chickens which $c$ pecks, and let $\mathcal{B}$ be the set of chickens in $c$'s flock, excluding $c$.
See Figure~\ref{fig:thm-any-chicken-pecked} for a depiction.
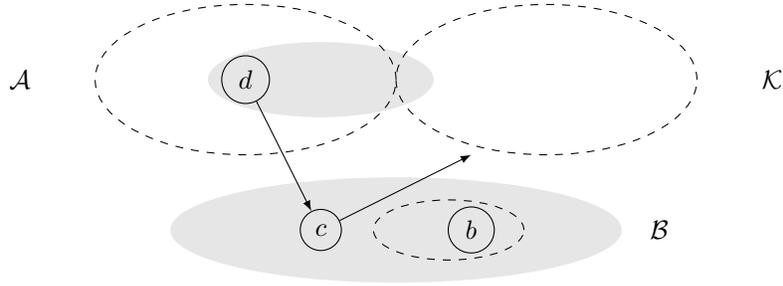
\begin{figure}[ht]
    \centering
    \begin{tikzpicture}
        
        \fill[black!10] (-1,0) ellipse (1.5 and 0.5) ;
        \fill[black!10] (0,-2) ellipse (3 and 0.7) ; 
        
        \draw[dashed] (-2,0) ellipse (2 and 1) ; 
        \draw[dashed] (0.7,-2) ellipse (1 and 0.4) ; 
        \draw[dashed] (2,0) ellipse (2 and 1) ; 
        
        \node[draw, circle] (d) at (-2,0) {$d$} ;
        \node[draw,circle] (c) at (-1,-2) {$c$} ;
        \node[draw, circle] (b) at (1,-2) {$b$} ;

        \node (A) at (-5,0) {$\mathcal{A}$} ;
        \node (B) at (3.5,-2) {$\mathcal{B}$} ;
        \node (K) at (5,0) {$\mathcal{K}$} ;
        
        \draw[-latex] (d) -- (c) ;
        \draw[-latex] (c) -- (1,-1) ;
        
    \end{tikzpicture}
    \caption{Initial set-up in Theorem~\ref{thm:any-chicken-pecked}.}
    \label{fig:thm-any-chicken-pecked}
\end{figure}
We know $d \rightarrow c$ and, consequently, $d$ is a 2-Duke over $\mathcal{K}$.
Consider an arbitrary chicken \(b \in \mathcal{B}\).
Observe that $d$ is a 3-Duke over $b$ if $b$ is pecked by any chicken in $\mathcal{K}$, or if $b$ is pecked by any chicken in $\mathcal{A}$ that $d$ is a 2-Duke over.

Suppose $d$ is not a 3-Duke in the overall graph.
Then there is some chicken $b' \in \mathcal{B}$ who pecks every chicken in $\mathcal{K}$ and every chicken in $\mathcal{A}$ that $d$ is a 2-Duke over.
See Figure~\ref{fig:thm-any-chicken-pecked2} for a depiction.
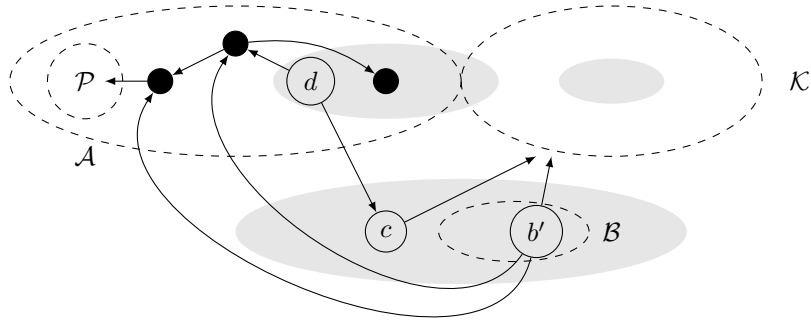
\begin{figure}[ht]
    \centering
    \begin{tikzpicture}
        
        \fill[black!10] (-1,0) ellipse (1.5 and 0.5) ;
        \fill[black!10] (0,-2) ellipse (3 and 0.7) ; 
        \fill[black!10] (2,0) ellipse (0.7 and 0.3) ;
        
        \draw[dashed] (-3,0) ellipse (3 and 1) ; 
        \draw[dashed] (0.7,-2) ellipse (1 and 0.4) ; 
        \draw[dashed] (2,0) ellipse (2 and 1) ; 
        \draw[dashed] (-5,0) ellipse (0.5 and 0.5) ; 
        
        \node[draw, circle] (d) at (-2,0) {$d$} ;
        \node[draw,circle] (c) at (-1,-2) {$c$} ;
        \node[draw, circle] (b) at (1,-2) {$b'$} ;
        
        \node[draw, circle, fill] (d1) at (-3,0.5) {} ;
        \node[draw, circle, fill] (d2) at (-4,0) {} ;
        
        \node (A) at (-5,-1) {$\mathcal{A}$} ;
        \node (B) at (2,-2) {$\mathcal{B}$} ;
        \node (K) at (4.5,0) {$\mathcal{K}$} ;
        \node (P) at (-5,0) {$\mathcal{P}$} ;
        
        \node[draw,fill,circle] (oops) at (-1,0) {} ;
        
        \draw[-latex] (d) -- (c) ;
        \draw[-latex] (c) -- (1,-1) ;
        \draw[-latex] (b) -- (1.2,-1) ;
        \draw[-latex] (d) -- (d1) ;
        \draw[-latex] (d1) -- (d2) ;
        \draw[-latex] (b) to [bend left=90] (d1) ;
        \draw[-latex] (b) to [bend left = 100] (d2) ;
        \draw[-latex] (d2) to (P) ;
        \draw[-latex] (d1) to [bend left = 20] (oops) ;
        
    \end{tikzpicture}
    \caption{Final situation in Theorem~\ref{thm:any-chicken-pecked}.}
    \label{fig:thm-any-chicken-pecked2}
\end{figure}
We claim $b'$ is a 2-Duke.
Let $\mathcal{P}$ be the set of chickens in $\mathcal{A}$ to which $d$ has peck chains entirely within \(\mathcal{A}\) of length exactly 3.
It remains to be seen that $b'$ is a 2-Duke over \(\mathcal{P}\) and \(F_d\). 
Since $b'$ pecks all the chickens who $d$ has peck chains of length at most 2 to, $b'$ has peck chains of length at most 2 to all the chickens in \(\mathcal{P}\).
As for the chickens in $F_d$, because $b'$ pecks everything $d$ pecks, if there were some chicken in $d$'s flock that $b'$ was not a 2-Duke over, it would have to peck $b'$ and all of the chickens $d$ pecks, which is a contradiction, because $d$ pecks the most chickens among all 3-Dukes in $\mathcal{A}$.
Therefore, as $b'$ has a peck chain of length $2$ or less to all other chickens not in its flock, $b'$ is a 2-Duke.
\end{proof}

This theorem is of particular interest, as 2-Dukes aren't guaranteed.
In fact, we can use this to immediately show a neat result!

\begin{corollary}
\label{cor:2d-or-4-3d}
In any multi-flock graph,
either there is a 2-Duke, or there are three 3-Dukes.
\end{corollary}

\begin{proof}
Suppose there is no 2-Duke.
Then there is no 1-Duke either.
By Theorem~\ref{thm:any-chicken-pecked}, any chicken pecked is pecked by a 3-Duke.
Consider a 3-Duke guaranteed by Theorem~\ref{thm:nflock-exists-3Duke}.
Because our 3-Duke must be pecked (or else it would be a 1-Duke), we know there must exist another 3-Duke which pecks it.
That 3-Duke cannot be pecked by our original 3-Duke, and therefore must be pecked by a third 3-Duke.
\end{proof}

Now we can move into proving a greater theory of Dukes, a more comprehensive map of how many Dukes of which types must always exist.
We already know that there exist graphs with only four 3-Dukes, and no other Dukes---Figure~\ref{fig:biflock-4-3ds}.
Let us show that there will always be four 3-Dukes; or else there exists at least one 2-Duke.
Why have this clause, you may wonder?
It is well shown by Figure~\ref{fig:4-flock-3-2-Dukes}, in which only three 2-Dukes, and no other Dukes, exist.
\begin{figure}[ht]
    \centering
    \begin{tikzpicture}
        \coordinate (fa) at (-2,2) ;
        \coordinate (fb) at (0.7,2) ;
        \coordinate (fc) at (2,1) ;
        \coordinate (e) at (0,-2) ;
    
        \fill[black!10] (fa) ellipse (0.5 and 0.3) ;
        \fill[black!10] (fb) ellipse (0.5 and 0.3) ;
        \fill[black!10] (fc) ellipse (0.5 and 0.3) ;
        \fill[black!10] (e) circle (1) ;
        
        \node[draw,circle,fill] (a) at (fa) {} ;
        \node[draw,circle,fill] (b) at (fb) {} ;
        \node[draw,circle,fill] (c) at (fc) {} ;
        
        \draw[-latex] (a) to [bend left=10] (b) ;
        \draw[-latex] (c) to (a) ;
        \draw[-latex] (b) to [bend left=10] (c) ;
        
        \draw[-latex] (a) to (-0.3,-1.05) ;
        \draw[-latex] (b) to (0.1,-1) ;
        \draw[-latex] (c) to (0.3,-1.05) ;

    \end{tikzpicture}
    \caption{A 4-flock graph with exactly three 2-Dukes, each pecking every chicken in the bottom flock.}
    \label{fig:4-flock-3-2-Dukes}
\end{figure}
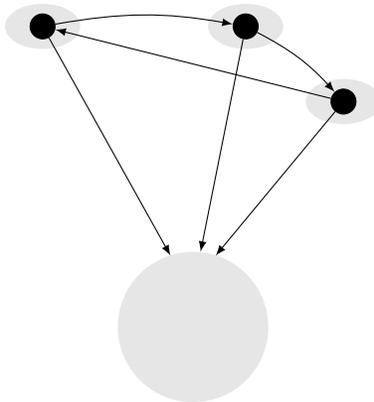
Now we can continue worKing our way up and showing the existence of more Dukes. 
In the case where there are no 2-Dukes, we have shown in Corollary~\ref{cor:2d-or-4-3d} that there exist three 3-Dukes.
All that remains is to show the existence of one more.

\begin{theorem}
\label{thm:exists-fourth-3Duke}
In a multi-flock graph with no 2-Dukes and at least three 3-Dukes, there must exist a fourth 3-Duke.
\end{theorem}

\begin{proof}
Suppose there is no 2-Duke and that there are at least three 3-Dukes, \(d_1, d_2\), and \(d_3\).
Suppose, for a contradiction that these are the only 3-Dukes in the graph.
Without loss of generality, we can assume \(d_1 \rightarrow d_2 \rightarrow d_3 \rightarrow d_1\) by Theorem~\ref{thm:any-chicken-pecked}.

Now consider the chickens other than \(d_3\) which peck $d_1$.
If no such chickens exist, then $d_1$ would be a 2-Duke, as any chicken not in its flock is either directly pecked by $d_1$, or is $d_3$, in which case we have $d_1 \rightarrow d_2 \rightarrow d_3$.
By assumption, this cannot happen.

This means there must be some chicken which pecks $d_1$ and is not one of our three original 3-Dukes.
Now let \(\mathcal{G}\) be the multi-flock graph obtained by removing \(d_3\).
By Theorem~\ref{thm:any-chicken-pecked}, either one of the chickens $c$ who pecks $d_1$ must be a 3-Duke in \(\mathcal{G}\), or $d_1$ shares a flock with a 2-Duke of \(\mathcal{G}\).

\underline{Case 1:} $c$ is a 3-Duke in \(\mathcal{G}\).
When we add $d_3$ back in,
because $c \rightarrow d_1 \rightarrow d_2 \rightarrow d_3$, we see that $c$ remains a 3-Duke, and we have found a fourth 3-Duke.

\underline{Case 2:} $d_1$ shares a flock with a 2-Duke in \(\mathcal{G}\).
Let $t$ be a 2-Duke of \(\mathcal{G}\) which is also in $d$'s flock.
Then $t$ has peck chains of length at most two to every chicken besides $d_3$.
When we add $d_3$ back in, $t$ will still have a peck chain of length at most 2 to $d_2$, and therefore peck chains of length at most 3 to $d_3$ and every other chicken.
Thus, $t$ is a fourth 3-Duke.
\end{proof}

\begin{corollary}
\label{cor:either-2D-or-3Ds}
In any multi-flock graph, either there is a 2-Duke, or there are four 3-Dukes.
\end{corollary}

\begin{proof}
We know by Corollary~\ref{cor:2d-or-4-3d} that either there exists a 2-Duke or three 3-Dukes.
By Theorem~\ref{thm:exists-fourth-3Duke}, if there is not 2-Duke, then four 3-Dukes exist.
Therefore, either there exists a 2-Duke or four 3-Dukes. 
\end{proof}

After all of these, we finally have a good base case to work from---either four 3-Dukes, or a 2-Duke. In fact, it is enough to move onto our longest proof. 
Just kidding! Let's first prove something about chickens eclipsing one another. 
Recall the definition of eclipses; a chicken eclipses another chicken in its flock if it pecks all chickens that the other chicken pecks, and at least one more. 

\begin{lemma}
Whenever there exists an $m$-Duke $d$ in some flock, there must exist a non-eclipsed $m$-Duke in that flock.
Additionally, if some chicken is pecked by an $m$-Duke $d$, it must also be pecked by a non-eclipsed $m$-Duke.
\end{lemma}

\begin{proof}
If our $m$-Duke $d$ is non-eclipsed, we are done.
Otherwise, there must be finitely many chickens which eclipse it.
Consider one such chicken $e$ which pecks the most chickens.
As any chicken which eclipsed $e$ would also eclipse $d$, there can be no chicken which eclipses $e$.
As $e$ pecks all the chickens $d$ pecks, it must be an $m$-Duke.
\end{proof}

Okay, now we can move on to the longest proof in the paper. It's got a lot of cases; the main proof is broken into cases 1, 2, and 3. Each is then broken into sub-cases (a, b, ...) which in turn may contain their own sub-cases (i, ii, ...). Without further ado, we'll move into

\begin{lemma}
\label{thm:2Duke-pecked-or-else}
Consider a multi-flock graph with a {non-eclipsed} 2-Duke, $d$.
Then at least one of the following happens.
\begin{enumerate}[(i)]
    \item A 1-Duke exists.
    \item The 2-Duke $d$ is pecked by another 2-Duke.
    \item Three 2-Dukes exist.
    \item Four 3-Dukes exist.
\end{enumerate}
\end{lemma}

\begin{proof}
Suppose there are no 1-Dukes, and consider this {non-eclipsed} 2-Duke, $d$.
Since $d$ is not a 1-Duke, there exist chickens which peck $d$.
Any chicken which pecks $d$ has a peck chain of length at most 3 to all of the multi-flock graph except the chickens in the $d$'s flock.

\underline{Case 1:} $d$ is pecked by at least 3 chickens.
Let $\mathcal{K}$ be the set of chickens which are both in $d$'s flock \emph{and} peck all of the chickens $d$ pecks.
The chickens in $\mathcal{K}$ must also be 2-Dukes, as they peck all of the chickens $d$ pecks, although none are guaranteed.
Every chicken which pecks $d$ must be a 3-Duke over all of the multi-flock graph except $d$'s flock.
However, by the definition of \(\mathcal{K}\), all of $d$'s flock but $\mathcal{K} \cup \{d\}$ is pecked by a chicken pecked by $d$.
Therefore, the chickens pecking \(d\) are only potentially not 3-Dukes over \(\mathcal{K}\).

If \(\mathcal{K}\) has size zero, then there are four 3-Dukes: $d$ and the three chickens which peck it.
If \(\mathcal{K}\) has size at least two, then there are (at least) three 2-Dukes: $d$ and the chickens in $\mathcal{K}$.
The remaining case in when \(\mathcal{K}\) contains exactly one chicken, \(k\).
In this case, no chicken in \(k\) and \(d\)'s flock---besides \(k\) and \(d\) themselves---can peck all the chickens that \(k\) pecks.
This implies \(k\) has a peck chain of length 2 to every chicken in its flock other than itself and \(d\).
Let $\mathcal{P}$ be the set of chickens which peck $k$,
and observe that the chickens in \(\mathcal{P}\) must also peck $d$.

In this case, all chickens in $\mathcal{P}$ are 3-Dukes.

\underline{Case 1(a):} $|\mathcal{P}| \geqslant 2$.
See Figure~\ref{fig:thm5-case1a} for a depiction.
\begin{figure}
    \centering
    \begin{tikzpicture}
    
        \fill[black!10] (0,0) ellipse (4 and 1) ;
        
        \draw[dashed] (2,0) ellipse (2 and 0.5) ; 
        \draw[dashed] (1.5,2) ellipse (2.75 and 0.5) ; 
            
        \node[draw, circle] (d) at (-1,0) {$d$} ;
        \node[draw, circle] (k) at (2,0) {$k$} ;
        
        \node[draw,circle,fill] (e) at (-2,-2) {} ;
        \node[draw,circle,fill] (f) at (-3,0) {} ;
        \node[draw,circle,fill] (d1) at (-1.5,2) {} ;
        \node[draw,circle,fill] (d2) at (-1.0,2) {} ;
        \node[draw,circle,fill] (d3) at (-0.5,2) {} ;
        
        \node (K) at (4.5,0) {$\mathcal{K}$} ;
        \node (P) at (4.5,2) {$\mathcal{P}$} ;
        
        \draw[-latex] (d1) to (d) ;
        \draw[-latex] (d2) to (d) ;
        \draw[-latex] (d3) to (d) ;
        \draw[-latex] (d) to (e) ;
        \draw[-latex] (e) to (f) ;
        \draw[-latex] (2,1.5) to (k) ;
        \draw[-latex] (1,1.5) to (d) ;
        \draw[-latex] (k) to (e) ;

    \end{tikzpicture}
    \caption{Case 1(a) of Lemma~\ref{thm:2Duke-pecked-or-else}.}
    \label{fig:thm5-case1a}
\end{figure}
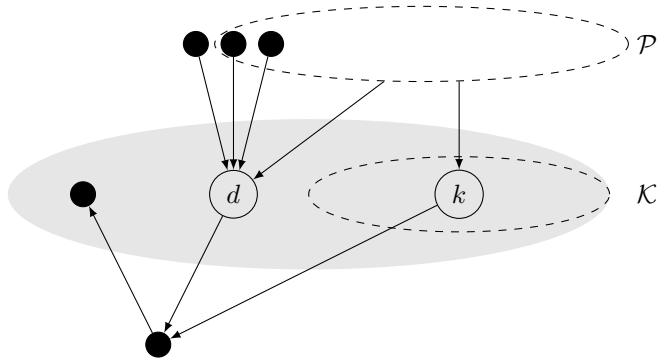
In this case, there are four 3-Dukes: $d$, $k$, and the chickens in $\mathcal{P}$.

\underline{Case 1(b):} $|\mathcal{P}| = 1$.
See Figure~\ref{fig:thm5-case1b} for a depiction.
\begin{figure}
    \centering
    \begin{tikzpicture}
    
        \fill[black!10] (0,0) ellipse (4 and 1) ;
        
        \draw[dashed] (2,0) ellipse (2 and 0.5) ; 
        \draw[dashed] (1.5,2) ellipse (2.25 and 0.5) ; 
            
        \node[draw, circle] (d) at (-1,0) {$d$} ;
        \node[draw, circle] (k) at (2,0) {$k$} ;
        
        \node[draw,circle,fill] (e) at (-2,-2) {} ;
        \node[draw,circle,fill] (f) at (-3,0) {} ;
        \node[draw,circle,fill] (d1) at (-1.5,2) {} ;
        \node[draw,circle,fill] (d2) at (-1.0,2) {} ;
        \node[draw,circle] (p) at (0,2) {$p$} ;
        
        \node (K) at (4.5,0) {$\mathcal{K}$} ;
        \node (P) at (4,2) {$\mathcal{P}$} ;
        
        \draw[-latex] (d1) to (d) ;
        \draw[-latex] (d2) to (d) ;
        \draw[-latex] (p) to (d) ;
        \draw[-latex] (d) to (e) ;
        \draw[-latex] (e) to (f) ;
        \draw[-latex] (p) to (k) ;
        \draw[-latex] (k) to (e) ;

    \end{tikzpicture}
    \caption{Case 1(b) of Lemma~\ref{thm:2Duke-pecked-or-else}.}
    \label{fig:thm5-case1b}
\end{figure}
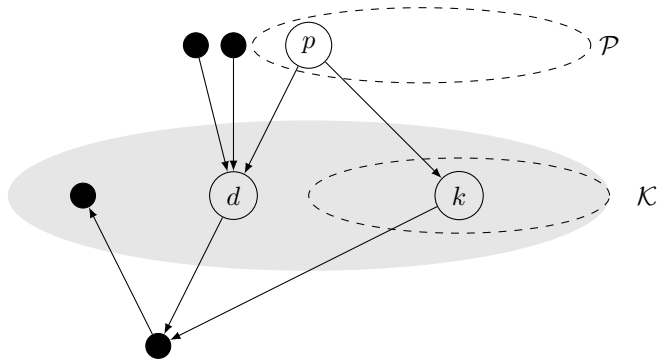
In this case, let $p$ be the sole member of $\mathcal{P}$, and note that \(k\) pecks every chicken not in its flock except for \(p\).
Then $p$ is a 2-Duke over all chickens not in $k$ and $d$'s flock.
Any chicken $c$ in $k$'s flock which pecked $p$ would be a 3-Duke.
We know that neither $k$ nor $d$ peck $p$, so if such a chicken $c$ exists, then there are four 3-Dukes: $k$, $d$, $p$, and $c$.
If $p$ pecks all chickens in $k$'s flock, then $p$ is a 2-Duke, and so there are three 2-Dukes: $k$, $d$, and $p$.


\underline{Case 2:} $d$ is pecked by exactly two chickens.
Let $i$ and $j$ be the two chickens which peck $d$.
In this case, either $i$ pecks $j$ without loss of generality, or they share a flock.

\underline{Case 2(a):} $i$ pecks $j$.
If $i$ pecks $j$, then $i$ is a 2-Duke over all chickens not in $d$'s flock, because $d$ pecks all chickens outside of its own flock other than $i$ and $j$.
Additionally, $j$ is a 3-Duke over the chickens not in $d$'s flock, because $d$ is a 2-Duke and $j \rightarrow d$.
We also know that any chicken in $d$'s flock which pecks $i$ is a 3-Duke.

\underline{Case 2(a)\textit{i}:} no chicken in $d$'s flock pecks both $i$ and $j$.
If no chicken in $d$'s flock pecks both $i$ and $j$, then $i$ is a 2-Duke which pecks $d$, and we're done.

\underline{Case 2(a)\textit{ii}:} some chicken, $e$, in $d$'s flock, pecks both $i$ and $j$.
See Figure~\ref{fig:thm5-case2aii} for a depiction.
\begin{figure}
    \centering
    \begin{tikzpicture}
        
        \fill[black!10] (-1,2) ellipse (0.8 and 0.6) ;
        \fill[black!10] (1,2) ellipse (0.8 and 0.6) ;
        \fill[black!10] (0,0) ellipse (3 and 1) ;
    
        \node[draw,circle] (d) at (0,0) {$d$} ;
        \node[draw,circle] (i) at (-1,2) {$i$} ;
        \node[draw,circle] (j) at (1,2) {$j$} ;
        \node[draw,circle] (e) at (-2,0) {$e$} ;
        
        \node[draw,circle,fill] (v1) at (-1,-2) {} ;
        \node[draw,circle,fill] (v2) at (1,-2) {} ;
        \node[draw,circle,fill] (v3) at (2,0) {} ;
        
        \draw[-latex] (i) to (d) ;
        \draw[-latex] (j) to (d) ;
        \draw[-latex] (i) to (j) ;
        \draw[-latex] (e) to (i) ;
        \draw[-latex] (e) to (j) ;
        \draw[-latex] (d) to (v1) ;
        \draw[-latex] (v1) to (e) ;
        \draw[-latex] (d) to (v2) ;
        \draw[-latex] (v2) to (v3) ;
    \end{tikzpicture}
    \caption{Case 2(a)\textit{ii} of Lemma~\ref{thm:2Duke-pecked-or-else}.}
    \label{fig:thm5-case2aii}
\end{figure}
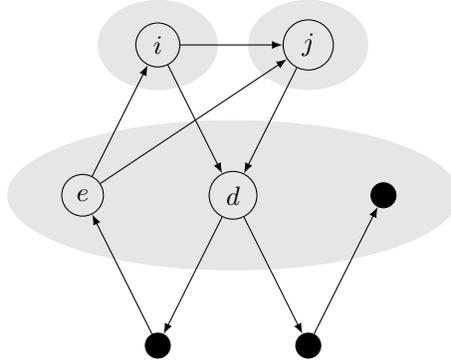
In this case, $e$ is a 3-Duke because $e \rightarrow i \rightarrow j$, $e \rightarrow i \rightarrow d$, and $d$ pecks every chicken not in its flock other than $i$ and $j$.
Furthermore, since $d$ is {non-eclipsed}, every other chicken in its flock is pecked by something $d$ pecks.
This implies $i$ and $j$ are 3-Dukes.
As $d$, $e$, $i$, and $j$ are 3-Dukes, we are done.


\underline{Case 2(b):} $i$ and $j$ share a flock.
If $i$ and $j$ share a flock, then they are both 2-Dukes over any chicken not in $d$'s flock, as $d$ pecks every chicken not in its flock other than $i$ and $j$.
Note that $i$ and $j$ must then also be 3-Dukes over any chicken in $d$'s flock which they either peck directly or which does not peck all chickens pecked by $d$.
See Figure~\ref{fig:thm5-case2b-i3d} for a depiction.
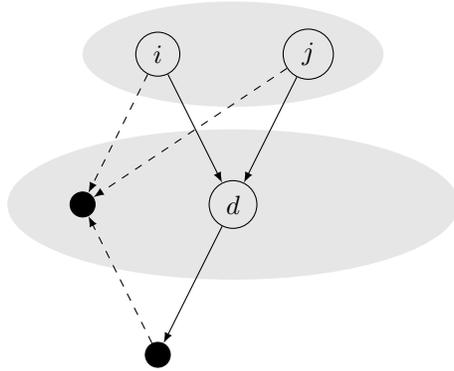
\begin{figure}
    \centering
    \begin{tikzpicture}
        
        \fill[black!10] (0,0) ellipse (3 and 1) ;
        \fill[black!10] (0,2) ellipse (2 and 0.7) ;
    
        \node[draw,circle] (d) at (0,0) {$d$} ;
        \node[draw,circle] (i) at (-1,2) {$i$} ;
        \node[draw,circle] (j) at (1,2) {$j$} ;
        
        \node[draw,circle,fill] (e) at (-1,-2) {};
        \node[draw,circle,fill] (f) at (-2,0) {};
        
        \draw[-latex] (i) to (d) ;
        \draw[-latex] (j) to (d) ;
        \draw[-latex] (d) to (e) ;
        \draw[-latex,dashed] (e) to (f) ;
        \draw[-latex,dashed] (i) to (f) ;
        \draw[-latex,dashed] (j) to (f) ;
    \end{tikzpicture}
    \caption{Possibility of $i$ or $j$ being a 3-Duke over a chicken in $d$'s flock from Case 2(b) of Lemma~\ref{thm:2Duke-pecked-or-else}.}
    \label{fig:thm5-case2b-i3d}
\end{figure}
Any chicken in $d$'s flock which pecks at least one of $i$ or $j$ and all of the chickens $d$ pecks must be a 2-Duke, and either at least one such chicken must exist, or $i$ and $j$ are both 3-Dukes.
See Figure~\ref{fig:thm5-case2b-2Duke} for a depiction.
\begin{figure}
    \centering
    \begin{tikzpicture}
        
        \fill[black!10] (0,0) ellipse (3 and 1) ;
        \fill[black!10] (0,2) ellipse (2 and 0.7) ;
    
        \node[draw,circle] (d) at (0,0) {$d$} ;
        \node[draw,circle] (i) at (-1,2) {$i$} ;
        \node[draw,circle] (j) at (1,2) {$j$} ;
        
        \node[draw,circle,fill] (e) at (-1,-2) {};
        \node[draw,circle,fill] (f) at (-2,0) {};
        
        \draw[-latex] (i) to (d) ;
        \draw[-latex] (j) to (d) ;
        \draw[-latex] (d) to (e) ;
        \draw[-latex] (f) to (e) ;
        \draw[-latex,dashed] (f) to (i) ;
        \draw[-latex,dashed] (f) to (j) ;
    \end{tikzpicture}
    \caption{Possibility of a second 2-Duke in $d$'s flock from Case 2(b) of Lemma~\ref{thm:2Duke-pecked-or-else}.}
    \label{fig:thm5-case2b-2Duke}
\end{figure}
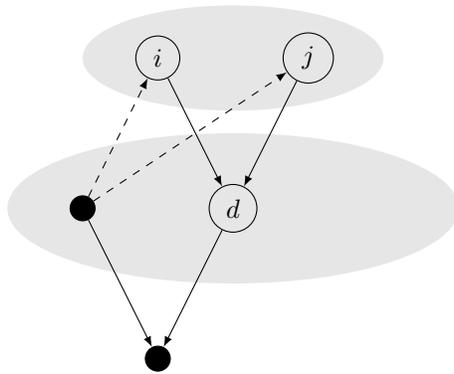

\underline{Case 2(b)\textit{i}:} $i$ and $j$ are both 3-Dukes.
In this case, by Theorem~\ref{thm:any-chicken-pecked}, for each of $i$ and $j$, it either shares its flock with a 2-Duke, or it is pecked by a 3-Duke.
If $i$ or $j$ is pecked by a 3-Duke, then we have four 3-Dukes: $d$, $i$, $j$, and the 3-Duke pecking $i$ or $j$.
If neither $i$ nor $j$ is pecked by a 3-Duke, then they both share their flock with a 2-Duke.
In the event that both $i$ and $j$ are 2-Dukes, then we have three 2-Dukes: $d$, $i$, and $j$.
Alternatively, if there is another 2-Duke in $i$ and $j$'s flock, then we have four 3-Dukes: $d$, $i$, $j$, and the other 2-Duke in $i$ and $j$'s flock.

\underline{Case 2(b)\textit{ii}:} There is some other 2-Duke $f$ in $d$'s flock which pecks either $i$ or $j$.
Note that $f$ cannot peck both $i$ and $j$, or else it would peck all chickens not in its flock, and be a 1-Duke.
Without loss of generality, suppose $i \rightarrow f \rightarrow j$, as in Figure~\ref{fig:thm5-case2bii-x}.
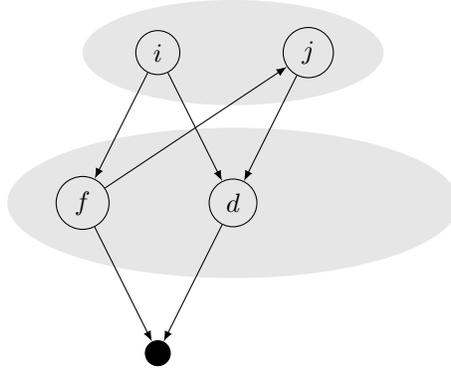
\begin{figure}
    \centering
    \begin{tikzpicture}
        
        \fill[black!10] (0,0) ellipse (3 and 1) ;
        \fill[black!10] (0,2) ellipse (2 and 0.7) ;
    
        \node[draw,circle] (d) at (0,0) {$d$} ;
        \node[draw,circle] (i) at (-1,2) {$i$} ;
        \node[draw,circle] (j) at (1,2) {$j$} ;
        
        \node[draw,circle,fill] (e) at (-1,-2) {};
        \node[draw,circle] (f) at (-2,0) {$f$};
        
        \draw[-latex] (i) to (d) ;
        \draw[-latex] (j) to (d) ;
        \draw[-latex] (d) to (e) ;
        \draw[-latex] (f) to (e) ;
        \draw[-latex] (i) to (f) ;
        \draw[-latex] (f) to (j) ;
    \end{tikzpicture}
    \caption{Case 2(b)\textit{ii} of Lemma~\ref{thm:2Duke-pecked-or-else}.}
    \label{fig:thm5-case2bii-x}
\end{figure}
If $f$ is the only other 2-Duke in $d$'s flock, then $i$ is a 3-Duke, which in turn by Theorem~\ref{thm:any-chicken-pecked} must be pecked by a 3-Duke or share a flock with a 2-Duke.
In the first case, there are four 3-Dukes: $d$, $f$, $i$, and the 3-Duke pecking $i$.
In the second case, there are three 2-Dukes: $d$, $f$, and the other 2-Duke in $i$'s flock.
However, if $f$ is not the only other 2-Duke in $d$'s flock, then we have three 2-Dukes: $d$, $f$, and yet another 2-Duke in their flock.

\underline{Case 3:} $d$ is pecked by exactly one chicken $t$.
In this case, $t$ is a 2-Duke over all chickens not in $d$'s flock
Furthermore, $t$ is a 3-Duke over all chickens because no chicken in $d$'s flock can peck all chickens $d$ pecks and $t$, or they would peck all chickens not in their flock, and be a 1-Duke.
Any chicken in $d$'s flock which pecks $t$ is a 3-Duke.

\underline{Case 3(a):} there are at least two chickens in $d$'s flock who peck $t$.
In this case, there are four 3-Dukes: $d$, $t$, and the multiple chickens in $d$'s flock who peck $t$.

\underline{Case 3(b):} there are no chickens in $d$'s flock who peck $t$.
In this case, $t$ is a 2-Duke pecking $d$, and we are done.

\underline{Case 3(c):} there is exactly one chicken, $r$, in $d$'s flock who pecks $t$.
In this case, we already have three 3-Dukes: $d$, $t$, and $r$.
Now we have all the tools we need to finish the proof.
To do this, we will consider any chicken $u$ which pecks $r$ and is not in $t$'s flock.

\underline{Case 3(c)\textit{i}:}
If such a chicken exists, either $t \rightarrow u$ or $u \rightarrow t$.
See Figure~\ref{fig:thm5-case3ci} for a depiction.
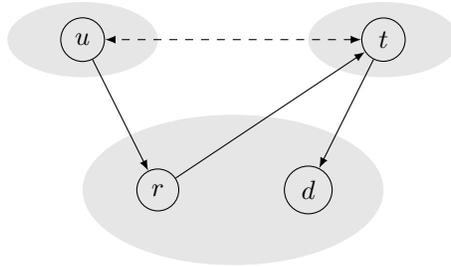
\begin{figure}
    \centering
    \begin{tikzpicture}
        \fill[black!10] (0,0) ellipse (2 and 1) ;
        \fill[black!10] (-2,2) ellipse (1 and 0.5) ;
        \fill[black!10] (2,2) ellipse (1 and 0.5) ;
        
        \node[draw,circle] (d) at (1,0) {$d$} ;
        \node[draw,circle] (t) at (2,2) {$t$} ;
        \node[draw,circle] (u) at (-2,2) {$u$} ;
        \node[draw,circle] (r) at (-1,0) {$r$} ;
        
        \draw[-latex] (t) to (d) ;
        \draw[-latex] (u) to (r) ;
        \draw[-latex] (r) to (t) ;
        
        \draw[latex-latex,dashed] (t) to (u) ;
    \end{tikzpicture}
    \caption{Case 3(c)\textit{i} of Lemma~\ref{thm:2Duke-pecked-or-else}.}
    \label{fig:thm5-case3ci}
\end{figure}
If $t \rightarrow u$, then $t$ is a 2-Duke, as $t \rightarrow d$, which pecks everything outside of its flock except $t$, and $t$ pecks every chicken in $d$'s flock except $r$, and $t \rightarrow u \rightarrow r$.
In this case, we are done, as $t$ is a 2-Duke which pecks $d$.
If, instead, $u \rightarrow t$, then $u$ is a 3-Duke, as
$u \rightarrow t \rightarrow d$ which pecks everything not in $d$'s flock,
$u \rightarrow t$ which pecks everything in $d$'s flock except $r$,
and
$u \rightarrow r$.
Furthermore, $r$ is a 3-Duke, as $r \rightarrow t \rightarrow d$ which pecks everything not in $d$ and $r$'s flock.
Then, we have four 3-Dukes: $t$, $d$, $u$, and $r$.

\underline{Case 3(c)\textit{ii}:}
The final case we must consider is if there is no chicken which both pecks $r$ and is not in $t$'s flock.
In other words, $r$ pecks every chicken not in $t$'s flock and not in $r$'s flock.
Since $r$ is not a 1-Duke, there is still some chicken $v$ in $t$'s flock which pecks $r$, as in Figure~\ref{fig:thm5-case3cii}.
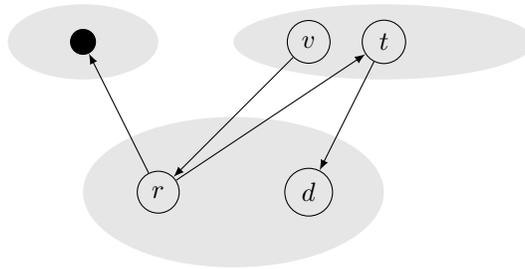
\begin{figure}
    \centering
    \begin{tikzpicture}
        \fill[black!10] (0,0) ellipse (2 and 1) ;
        \fill[black!10] (-2,2) ellipse (1 and 0.5) ;
        \fill[black!10] (2,2) ellipse (2 and 0.5) ;
        
        \node[draw,circle] (d) at (1,0) {$d$} ;
        \node[draw,circle] (t) at (2,2) {$t$} ;
        \node[draw,circle,fill] (u) at (-2,2) {} ;
        \node[draw,circle] (r) at (-1,0) {$r$} ;
        \node[draw,circle] (v) at (1,2) {$v$} ;
        
        \draw[-latex] (t) to (d) ;
        \draw[-latex] (r) to (u) ;
        \draw[-latex] (r) to (t) ;
        \draw[-latex] (v) to (r) ;
        
    \end{tikzpicture}
    \caption{Case 3(c)\textit{ii} of Lemma~\ref{thm:2Duke-pecked-or-else}.}
    \label{fig:thm5-case3cii}
\end{figure}
As in the previous case, $d$, $t$, and $r$ are still 3-Dukes.
To complete the proof, it remains to show that $v$ is a 3-Duke as well.
Observe that
$v \rightarrow r$ which pecks every chicken not in $v$'s flock or $r$'s flock,
and
$v \rightarrow r \rightarrow t$ which pecks every chicken in $r$'s flock except for $r$
Therefore, $v$ has a chain of length at most three to any chicken not in its flock and again we have four 3-Dukes: $d$, $t$, $r$, and $v$.
\end{proof}

\begin{theorem}
\label{thm:multiflock-finale}
In any multi-flock graph, there exists either a 1-Duke, three 2-Dukes, or four 3-Dukes
\end{theorem}

\begin{proof}
By Corollary~\ref{cor:either-2D-or-3Ds}, there exists either a 2-Duke or four 3-Dukes.
If there are four 3-Dukes, we're done.
Otherwise, by Lemma~\ref{thm:2Duke-pecked-or-else} and Lemma 3, we can consider a non-eclipsed 2-Duke and conclude that either it is pecked by a 2-Duke, a 1-Duke exists, three 2-Dukes exist, or four 3-Dukes exist.
We have completed our proof unless our original 2-Duke is pecked by a 2-Duke.
Our final step is to consider a non-eclipsed such 2-Duke, as by Lemma 3 one must exist, and apply Lemma~\ref{thm:2Duke-pecked-or-else} again.
Either our second 2-Duke is also pecked by a 2-Duke, in which case there are three 2-Dukes, or a 1-Duke exists, or three 2-Dukes exist, or four 3-Dukes exist.
In any of these cases, our statement holds.
\end{proof}

Theorem~\ref{thm:multiflock-finale} is powerful in the sense that we cannot guarantee more than four 3-Dukes in the absence of other Dukes.
This is because we can show that sometimes only four 3-Dukes exist; figure~\ref{fig:5-flock-4-3-Dukes} depicts a graph in which we have only four 3-Dukes.
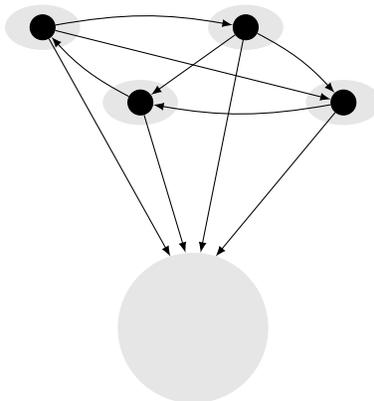
\begin{figure}[ht]
    \centering
    \begin{tikzpicture}
        \coordinate (fd) at (-0.7,1) ;
        \coordinate (fa) at (-2,2) ;
        \coordinate (fb) at (0.7,2) ;
        \coordinate (fc) at (2,1) ;
        \coordinate (e) at (0,-2) ;
    
        \fill[black!10] (fa) ellipse (0.5 and 0.3) ;
        \fill[black!10] (fb) ellipse (0.5 and 0.3) ;
        \fill[black!10] (fc) ellipse (0.5 and 0.3) ;
        \fill[black!10] (fd) ellipse (0.5 and 0.3) ;
        \fill[black!10] (e) circle (1) ;
        
        \node[draw,circle,fill] (a) at (fa) {} ;
        \node[draw,circle,fill] (b) at (fb) {} ;
        \node[draw,circle,fill] (c) at (fc) {} ;
        \node[draw,circle,fill] (d) at (fd) {} ;
        
        \draw[-latex] (a) to [bend left=10] (b) ;
        \draw[-latex] (a) to (c) ;
        \draw[-latex] (b) to [bend left=10] (c) ;
        \draw[-latex] (b) to (d) ;
        \draw[-latex] (c) to [bend left=10] (d) ;
        \draw[-latex] (d) to [bend left=10] (a) ;
        
        \draw[-latex] (a) to (-0.3,-1.05) ;
        \draw[-latex] (b) to (0.1,-1) ;
        \draw[-latex] (c) to (0.3,-1.05) ;
        \draw[-latex] (d) to (-0.1,-1) ;

    \end{tikzpicture}
    \caption{A 5-flock graph with exactly four 3-Dukes, each pecking every chicken in the bottom flock.}
    \label{fig:5-flock-4-3-Dukes}
\end{figure}

\section{Conclusion}

At this point, we may put down our pen.
Several basic existences of Dukes have been proved.
We have not examined 4-Dukes, or higher level Dukes, but in some sense there is not a need to.
Our goal was to examine the existence of a dominant chicken, and we have done so by proving the existence of 1-, 2-, and 3-Dukes
Thus, 4-Dukes, and any higher level Dukes, would not actually be dominant.

\section{Acknowledgements}

I'd like to thank Graham Gordon, my math mentor, who helped me turn my proofs into a paper.
I'd like to thank Mia Smith, my graph theory teacher, who set me on the path of exploration and proofread my first theorems.
Lastly, I'd like to thank Proof School, the institution which helped me connect with these people and encouraged me through all of it.

\bibliographystyle{alpha}
\bibliography{the}

\end{document}